\documentclass[runningheads]{llncs}
\usepackage[T1]{fontenc}
\usepackage{graphicx}
\usepackage{hyperref}
\usepackage{color}

\usepackage{amsmath,amssymb,amsfonts}

\renewcommand{\d}{\mathrm{d}}

\begin{document}
\title{Symplectic Approach to Global Stability}
%
%
\author{Verónica Errasti Díez \inst{1}\orcidID{0000-0003-0376-7402} \and
Jordi Gaset Rifà\inst{2,3}\orcidID{0000-0001-8796-3149} \and
Manuel Lainz Valcázar\inst{3}\orcidID{0000-0002-2368-5853}}
\authorrunning{V. Errasti Díez et al.}
%
\institute{Department of Mathematics,
CUNEF Universidad, 
Calle Pirineos 55,
28040 Madrid,
Spain 
\\
\email{\{veronica.errasti, jordi.gaset, manuel.lainz\}@cunef.edu}}
\maketitle              
\begin{abstract}
We present a new approach to the problem of proving global stability, based on symplectic geometry and with a focus on systems with several conserved quantities. 
We also provide a proof of instability for integrable systems whose momentum map is everywhere regular. 
Our results take root in the recently proposed  notion of a confining function
and are motivated by ghost-ridden systems, for whom
we put forward the first geometric definition.

\end{abstract}
%
%
%
\section{Motivation}

Global stability refers to stability properties that are satisfied by {\it all} solutions of a system. This feature is to be understood in contrast to local stability, where the focus is on the dynamics around particular points or surfaces. 

Global stability enjoys a lengthy history,
dating back to Lagrange. In its apogee in the 50s and 60s, it was studied for engineering applications wherein the variables do not need to be kept within stringent tolerances; see~\cite{Gyfto} and references therein. In such situations, local Lyapunov stability is too restrictive. Recently, global stability has seen a revived interest from the hand of theoretical physics. Motivated by the latter, we re-examine this notion, adapting it to the language of symplectic geometry.

Specifically, global stability is cornerstone to so-called ghost-ridden systems. These are a class of Hamiltonian systems that aim to simultaneously account for dark energy and quantum gravity, e.g.~\cite{Carroll:2003st,Hawking:2001yt}
and whose physical viability critically depends on
being proven globally stable.

The paper is structured as follows. In section \ref{sec:confining}, we adapt the two main notions of global stability in the physics literature to the dynamics generated by a vector field. We also review the allied concept of a confining function. In section \ref{sec:sym}, we present new results to prove global stability using several conserved quantities and the momentum map. In section \ref{sec:ghost}, we formalize ghost-ridden systems and point at bi-Hamiltonian theory as the befitting framework for their study. 
We draw our conclusions in the final section \ref{sec:conclusions}.

\section{Global Stability for a Vector Field}\label{sec:confining}

Let $M$ be a finite-dimensional smooth manifold in the sense of~\cite{Lee}. Assume that the dynamics is given by a  vector field $X$ on $M$. We restrict attention to two physically relevant global stability notions~\cite{ErrastiDiez:2024hfq},
which can be defined in terms of the integral curves $\gamma$ of $X$. We denote by $I_\gamma\subset \mathbb{R}$ the interval of definition of $\gamma$. We abide by the convention that a set $A$ is relatively compact with respect to a topological space $N$ if $A$ is contained in a compact set of $N$. 

\begin{definition}\label{def:GS} A vector field $X$ on $M$ is:
\begin{description}
    \item[G1 stable] if the image $Im(\gamma)$ is relatively compact with respect to $M$ for all integral curves $\gamma$. This is also known as Lagrange stability.
    \item[G2 stable] if, for all integral  curves $\gamma$, and for all subsets $S\subset I_\gamma$ relatively compact with respect to $\mathbb{R}$, its image $\gamma(S)$ is relatively compact with respect to $M$. This corresponds to $X$ being a complete vector field.
\end{description}
\end{definition}

A non-complete vector field has an integral curve that cannot be defined for all times. This implies that the image of the curve cannot be contained in any compact set of $M$~\cite{Lee}.
Then, we say that $X$ has a {\it blow up}. In particular, $G1$ implies $G2$.

{ 
To illustrate these definitions, consider the manifold $\mathbb{R}^2$ with coordinates $(q,v)$, and the vector fields:
\begin{align}
    X_1=v\frac{\partial}{\partial q}-q\frac{\partial}{\partial v}, \qquad  X_2=v\frac{\partial}{\partial q}, \qquad  X_3=v\frac{\partial}{\partial q}+q^2\frac{\partial}{\partial v}\,.
\end{align}

The first vector field $X_1$ corresponds to the harmonic oscillator and it is $G1$ stable. The second vector field $X_2$ corresponds to the free particle, and it is $G2$ stable but not $G1$. The last vector field $X_3$ is not $G2$ stable because it has at least an integral curve $\gamma(t)=(6t^{-2},-12t^{-3})$ that cannot be defined for all times. In references \cite{Damour:2021fva,Deffayet:2023wdg,ErrastiDiez:2024hfq}, the reader can find additional families of dynamical systems whose global stability has been determined.
}

A classical result states that, if $X$ is compactly supported (for instance, if $M$ is compact), then $X$ is complete~\cite{Lee}. Under this hypothesis, $X$ is straightforwardly $G1$ stable as well. Another classical result states that, if $X$ leaves invariant a proper map, then $X$ is complete. Recall that a proper map is a map $f:M\rightarrow N$ between topological spaces such that the anti-image of compact sets is compact~\cite{Lee:topo}. The latter result has recently been generalized through the notion of a confining function~\cite{ErrastiDiez:2024hfq}.

\begin{definition}\label{def:confining}
Let $M$ be a Hausdorff space
and let $\mathcal{S}$ be a set.
A function $f:M\rightarrow \mathcal{S}$
is \textbf{confining}
if the path-connected components of its level sets
are relatively compact with respect to $M$.
\end{definition}

Let $f^{{-1}}(\mu):=\{x\in M| f(x)=\mu\}\,$.
If $M$ is Hausdorff, a proper map $f:M\rightarrow N$ is confining~\cite{ErrastiDiez:2024hfq}. The converse is not true. For instance, $f:\mathbb{R}^2\rightarrow\mathbb{R}$ given by $f(x,y)=\sin(x^2+y^2)$ is confining but not a proper map.  If $M$ and $N$ are manifolds, $f$ is a smooth function and $\mu$ is a regular point of $N$, then the level set $f^{{-1}}(\mu)$ 
is a closed manifold. As such, its connected components are path-connected. We have the following theorem (adapted from~\cite{ErrastiDiez:2024hfq}):

\begin{theorem}\label{thm:confining}
If there exists a twice differentiable confining function $Q:M\rightarrow N$ such that $X(Q)=0$, then $X$ is $G1$ stable. In particular, it is complete.
\end{theorem}
\begin{proof}
The condition $X(Q)=0$ means that $Q$ is a conserved quantity of $X$. { Therefore, each integral curve of $X$ lies inside one path-connected component of a level set of $Q$, all of which are relatively compact with respect to $M$. Thus, all integral curves of $X$ are relatively compact with respect to $M$.}
\end{proof}

The converse of theorem \ref{thm:confining} is not true. Namely, there exist $G1$ stable systems without non-trivial conserved quantities. An example is the Kronecker foliation of the torus with irrational slope. In this case, $M$ is compact and so $X$ is $G1$ stable. However, each integral curve is dense in the torus and thus the only possible conserved quantities are constants.

\section{Extension to Several Symmetries }\label{sec:sym}

Let $M$ be a manifold with dimension $2n$. Consider the Hamiltonian system $(\omega, H)$, where $\omega\in\Omega^2(M)$ is a symplectic form and $H$ is a smooth function. Let $X_H$ be the Hamiltonian vector field associated to $(\omega,H)$, { that is \cite{AM-78}
\begin{align}
    i_{X_{H}}\omega=\d H\,.
\end{align}}
In particular, $H$ is a conserved quantity. 

The standard route to proving $G1$ stability in such a system by means of theorem \ref{thm:confining} is as follows. First, check if $H$ is confining.
If not, check if other conserved quantities exist
that are confining. 
In the event that several known conserved quantities exist, none of which is individually confining,
check if a combination thereof is confining.
For $n$ conserved quantities, it is further possible to use integrability techniques to (dis)prove $G1$ stability.
In the following, we show apposite results.

\subsection{$G1$ Stability for Systems with Multiple Conserved Quantities}\label{subsec:severalCQ}

Consider a Hamiltonian symplectic group action $\Phi:G\times M\rightarrow M$ of a Lie group $G$ on $(\omega,H)$. Let $J:M\rightarrow \mathfrak{g}^*$ be a momentum map for this action. 

\begin{proposition}\label{propJ}
    If $J$ is confining, then $X_H$ is $G1$ stable.
\end{proposition}
\begin{proof}
    If $H$ is invariant under $\Phi$, then the flow $F_t$ of $X_H$ leaves invariant $J$ \cite{AM-78} in the sense that, for all $x\in M$,
    \begin{align}
        J(F_t(x))=J(x)\,.
    \end{align}
    Therefore, each integral curve of $X_H$ is inside one path-connected component of a level set of $J$. Hence, they are all relatively compact with respect to $M$.
\end{proof}
{ 
Given $r$ functions $Q_i:M\rightarrow \mathbb{R}$, for $i=1,\dots,r$ on $M$, and a map $g:\mathbb{R}^r\rightarrow V$, the functional combination of $\{Q_i\}_{i=1,\dots,r}$ by $g$ is the map $g(Q_1,\dots,Q_r):M\rightarrow V$. If $Q_i$ are conserved quantities, $V$ is a manifold and $g$ is also at least once differentiable, then $g(Q_1,\dots,Q_r)$ is another conserved quantity. 

Consequently, it is reasonable to ask if a functional combination of the conserved quantities is confining. 
}As it turns out, if the momentum map $J$ is not confining (proper), then no functional combination of the conserved quantities is confining (proper).

\begin{lemma}\label{prop:composition}
  Let $M$ be a Hausdorff space and let $N$ and $V$ be two sets. If $f:M\rightarrow N$ is not confining, then for any function $g:N\rightarrow V$, $g\circ f$ is not confining.
\end{lemma}
\begin{proof}
If $f$ is not confining, then there exists a point $c\in N$ such that a path-connected component $F$ of $f^{{-1}}(c)$ cannot be contained in any compact set of $M$. Since $c\in g^{{-1}}(g(c))$, we have that $(g\circ f)^{{-1}}(g(c))\supset f^{{-1}}(c)$. In particular, it has a path connected component $G\subset (g\circ f)^{{-1}}(g(c))$ that contains $F$. Thus, $G$ cannot be contained in any compact set of $M$.
\end{proof}

\begin{lemma}\label{lem:proper}
Let $N$ be a Hausdorff space and let $M$ and $V$ be two topological spaces. If $f:M\rightarrow N$ is continuous but not a proper map, then for any continuous function $g:N\rightarrow V$, $g\circ f$ is not a proper map.
\end{lemma}

\begin{proof}
 If $f$ is not a proper map, then there exists a compact set $K\subset N$ such that $f^{{-1}}(K)$ is not compact, even though it is closed because $f$ is continuous and $N$ is Hausdorff. Since $g$ is continuous, $g(K)$ is compact. Moreover, $(g\circ f)^{-1}(g(K))\supset f^{{-1}}(K)$. If $(g\circ f)^{-1}(g(K))$ were compact, then $f^{{-1}}(K)$ would be compact, which is a contradiction.
\end{proof}

The converse of lemma \ref{lem:proper} holds true in the following case. Assume that the Hamiltonian system $(\omega,H)$ has conserved quantities $\{Q_i\}_{i=1,\dots ,r\leq n}$. Set $J:M\rightarrow \mathbb{R}^r$ as $J=(Q_1,\dots,Q_r)$ and define the function $P:M\rightarrow\mathbb{R}$ as $P:=\displaystyle \sum_{i=1}^rQ_i^2$.
        
\begin{proposition}{  The map $J=(Q_1,\dots,Q_r)$ is a proper map if, and only if, $P$ is a proper map.    }
\end{proposition}
\begin{proof}
Let $I\subset\mathbb{R}$ be a compact set with maximum $b\in\mathbb{R}$. Define $K:=[-\sqrt{|b|},\sqrt{|b|}]$ and let $K^r$ denote the Cartesian product of $r$ copies of $K$. Then, $Q_i(P^{{-1}}(I))\subset K$ for all $i$. Therefore, $
        P^{{-1}}(I)\subset\bigcap_{i=1}^rQ_i^{{-1}}(K)$. Moreover,
    \begin{align}
        J^{{-1}}(K^r)=(Q_1,\dots,Q_r)^{{-1}}(K^r)=\bigcap_{i=1}^rQ_i^{{-1}}(K)\,
    \end{align}
    is a compact set because $J$ is a proper map and $K^r$ is a compact set. Therefore, $P^{{-1}}(I)$ is inside a compact set. Besides, it is closed because $P$ is a continuous function. Thus, $P^{{-1}}(I)$ is a compact set.

    Conversely, if $J$ is not proper, then $P$ is not proper by lemma \ref{lem:proper}\,.
\end{proof}

\subsection{Global Stability for Integrable Systems}\label{subsec:unstable}

Consider a Hamiltonian system $(\omega,H)$ with $n$ conserved quantities $\{Q_i\}_{i=1,\dots n}$ in involution $\{Q_i,Q_j\}=0$ such that $\d Q_i$ are generically independent. { The bracket $\{\cdot,\cdot\}$ is the Poisson bracket defined by the symplectic form $\omega$ }. This implies that $H$ is a functional combination of the $Q_i$'s. The corresponding Hamiltonian vector fields, or symmetries, { are the unique vector fields $X_{Q_i}$ satisfying }
\begin{align}
    i_{X_{Q_i}}\omega=\d Q_i\,.
\end{align}
\noindent
Notice that such a system is not necessarily integrable because the symmetries may not be complete.

A geometric characterization is provided by a map $F:M\rightarrow \mathbb{R}^n$ that is generically a fibration with Lagrangian leaves on which the Hamiltonian is constant. We can be on one of the following three situations.

First, $F$ is confining. This is the standard Liouville-Mineur-Arnold theorem case. The  connected components of generic level sets of $F$ are tori and action-angle variables are defined on a neighbourhood of these tori. Thus, the integral curves are trapped inside tori and the system is $G1$ stable.

Second, the vector fields $\{X_{Q_i}\}_{i=1,\dots,n}$ generate a transitive action on the connected components of the level sets of $F$. The Liouville-Mineur-Arnold theorem for this case \cite{AM-78} states that the connected components of the level sets of regular points are not tori, but cylinders $\mathbb{T}^k\times\mathbb{R}^{n-k}$, with $k=0,\dots,n$. Moreover, on a neighbourhood of each cylinder there exist action-angle coordinates $(\theta^a, s^A,I_a,\bar{I}_A)$, with $a=1,\dots,k$ and $A=k+1,\dots,n$, such that $\theta^a$ are position coordinates on $\mathbb{T}^k$, $s^A$ are position coordinates on $\mathbb{R}^{n-k}$ and $(I_a,\bar{I}_A)$ are the respective momenta \cite{FGS-2003}. Furthermore, the equations of motion are
    \begin{align}
    \dot{\theta^a}&=\frac{\partial H (I_a,\bar{I}_A)}{\partial I_a}=:\omega_a(I_a,\bar{I}_A)\,,\quad  \dot{I_a}=0\,,
    \\
    \dot{s}^A&=\frac{\partial H(I_a,\bar{I}_A)}{\partial \bar{I}_A}=:\bar{\omega}_A(I_a,\bar{I}_A)\,,\quad  \dot{\bar{I}}_A=0\,.
\end{align}
Thus, the dynamics of $X_H$ is that of a uniform rectilinear motion on each cylinder. Note that the existence of these coordinates guarantees the 
absence of blow ups. Hence, the system is $G2$ stable.

If $\bar{\omega}_A(I_a,\bar{I}_A)$ is not identically zero on some cylinder, then the coordinate $s^A(t)$ will take arbitrarily large and small values. Hence, the system is $G2$ stable on the cylinders but not $G1$ stable. 

If $\bar{\omega}_A(I_a,\bar{I}_A)$ is identically zero for all $A$ and every cylinder, then the system does not move along the $\mathbb{R}^k$ directions. That is, the dynamics is confined to the tori $\mathbb{T}^{n-k}$ inside each cylinder. Therefore, the system is $G1$ stable on the cylinders. Moreover, in this situation, $s^A$ are also conserved quantities. Furthermore, the map 
\begin{align}
    (Q_i,\dots,Q_n,s^{k+1},\dots,s^n):M\rightarrow \mathbb{R}^{2n-k}
\end{align}
is confining because the connected components of its level sets are tori, which are compact.

Third, the flow of the Hamiltonian vector fields $X_{Q_i}$ are not all complete. In this case, the Liouville-Mineur-Arnold theorem does not apply and blow ups can occur.

The best scenario can be summarized in the following result.

\begin{theorem}\label{thm:iff} Let $(\omega, H)$ be a Hamiltonian system with conserved quantities $\{Q_i\}_{i=1,\dots,n}$ in involution and independent everywhere. Assume that the Hamiltonian vector fields $X_{Q_i}$ are complete on the common level sets of the $Q_i$'s. Then, the system is $G1$ stable if, and only if, there exists a smooth map $F:M\rightarrow\mathbb{R}^m$, with $m\geq n$, that is invariant by the dynamics and  confining.    
\end{theorem}

\section{An application to physics: ghost stabilization}\label{sec:ghost}

During the last three decades, there has been renewed and increased interest in a class of dynamical systems dubbed {\it ghost-ridden}. Their relevance lies in their twofold nature as candidates for dark energy and quantum gravity, e.g.~\cite{Carroll:2003st,Hawking:2001yt}. In spite of their wide use in theoretical physics, ghost-ridden systems lack a formal definition. In observance of the vast literature and following the comments of Smilga~\cite{Smilga:2008pr} on the work by Bender and Mannheim~\cite{Bender:2007wu} on the Pais-Uhlenbeck oscillator~\cite{PU}, we put forward the following formal definition. 

\begin{definition}
\label{def:ghost}
   {  A vector field $X$ on a manifold $M$ is ghost-ridden if, for all Hamiltonian systems $(\omega, H)$ on a manifold $M$ such that $i_X\omega=\d H$, $H$ is neither bounded from above nor from below.}
\end{definition}

A critical property that ghost-ridden systems must possess for physical viability is global stability.
Hence, ongoing efforts concentrate on the quest for $G1$ and/or $G2$ stable ghost-ridden systems. 
A prominent strategy propounds the construction
of systems with a completely unbounded Hamiltonian whose dynamics are globally stable.
Then, the primary difficulty consists in the ubiquitous emergence
of bi-Hamiltonian systems that do not fulfil definition \ref{def:ghost}, in the sense that only some of their Hamiltonians are completely unbounded, e.g.~\cite{Damour:2021fva,Kaparulin:2014vpa}.

To date, no globally stable system has been proven ghost-ridden. An enticing candidate was proposed in~\cite{ErrastiDiez:2024hfq}, where neither the Hamiltonian nor the confining conserved quantity that warrants $G1$ stability are bounded from above or below. Another candidate
is a certain modification of the Korteweg-de Vries partial differential equations~\cite{Smilga:2020elp}.
While global stability has not been proven,
there exist auspicious partial results in this regard~\cite{Fring:2024brg}.
In neither example, no systematic study of the underlying bi-Hamiltonian structure has been performed to confirm the ghost-ridden nature.
\newpage

\section{Summary and outlook}\label{sec:conclusions}

We have adapted the concepts of $G1$ and $G2$ global stability, as well as that of a confining function, to the context of the dynamics generated by a vector field. See definitions \ref{def:GS} and \ref{def:confining}. For the study of global stability, confining functions seem more adequate than proper maps.

We have presented new results on the global stability of systems with several known conserved quantities. We have shown that, if the momentum map is confining, then the system is $G1$ stable (proposition \ref{propJ}). Moreover, if the momentum map is not confining, then no functional combination of the conserved quantities can be confining (lemma \ref{lem:proper}).  

We have provided a more exhaustive analysis in the fortunate case that, on a $2n$-dimensional symplectic manifold, at least $n$ independent conserved quantities in involution are known. We have shown that, for complete symmetries on the level sets of regular points, the integral curves are inside a compact set if, and only if, there exists an invariant confining map. If all points are regular and symmetries are complete, then the result holds globally, see theorem \ref{thm:iff}.

In the theory of integrable systems, it is common to assume as a hypothesis that the Hamiltonian vector fields are complete, or that a certain foliation is compact. Either of these are
unwarranted assumptions in the problem of present interest, whose crucial aim is to prove such properties. Nevertheless, the language and objects of this theory provide an adequate starting point. In particular, it will be insightful to generalize the Liouville-Mineur-Arnold theorem to the case where the symmetries are not complete. This is the third case listed in section \ref{subsec:unstable}, which comprises work in progress.

We have put forward a formal definition \ref{def:ghost} of ghost-ridden systems  that reveals bi-Hamiltonian structures as the natural framework for their investigation. The endeavor for globally stable ghost-ridden systems is producing a plethora of bi-Hamiltonian systems, e.g.~\cite{ErrastiDiez:2024hfq,Kaparulin:2014vpa,Smilga:2008pr}, whose ghostly nature and/or global stability remain obscure. The use of bi-Hamiltonian techniques could shed light on the subject, which has never been explored from a geometric point of view, to the best of our knowledge.

Thus far, global stability proofs hinge on the existence of conserved quantities with well-disposed properties. However, this is not required by definition.
For three examples for which there exists numerical evidence in favor of global stability
without any confining conserved quantity, see~\cite{Deffayet:2023wdg}. The development of a technique to prove global stability without appealing for conserved quantities would entail a paradigm shift in ghost-ridden systems and comprises a compelling line of future research.

\begin{credits}
\subsubsection{\ackname} 
J.G.R. acknowledges financial support from Ministerio de Ciencia, Innovaci\'on y Universidades (Spain), PID2021-125515NB-21.

M. L. acknowledges financial support from Ministerio de Ciencia, Innovaci\'on y Universidades (Spain), PID2022-137909NB-C21.

\subsubsection{\discintname}
The authors have no competing interests to declare that are
relevant to the content of this article. 
\end{credits}
%
%
%

\newpage

\begin{thebibliography}{8}

\bibitem{AM-78}
Abraham,~R., Marsden,~J.~E.: Foundations of Mechanics. 2nd edn.
Benjamin--Cummings, Redwood City CA (1987).
(ISBN-10: 080530102X).

\bibitem{Bender:2007wu}
Bender, C.M., Mannheim, P.D.: No-ghost theorem for the fourth-order derivative Pais-Uhlenbeck oscillator model.
Phys. Rev. Lett. \textbf{100}, 110402 (2008)
\href{https://doi.org/10.1103/PhysRevLett.100.110402}
{doi:10.1103/PhysRevLett.100.110402}

\bibitem{Brouzet}
Brouzet, R., Caboz, R., Rabenivo, J., Ravoson, V.: Two degrees of freedom quasi-bi-Hamiltonian systems.
 J. Phys. A: Math. Gen. \textbf{29}, 2069 (1996)
\href{https://doi.org/10.1088/0305-4470/29/9/019}{doi:10.1088/0305-4470/29/9/019}

\bibitem{Carroll:2003st}
Carroll, S.M., Hoffman, M., Trodden, M.: Can the dark energy equation-of-state parameter $w$ be less than $-1$?.
Phys. Rev. D \textbf{68}, 023509 (2003)
\href{https://doi.org/10.1103/PhysRevD.68.023509}
{doi:10.1103/PhysRevD.68.023509}

\bibitem{Damour:2021fva}
Damour, T., Smilga, A.: Dynamical systems with benign ghosts.
Phys. Rev. D \textbf{105}, no.4, 045018 (2022)
\href{https://doi.org/10.1103/PhysRevD.105.045018}{
doi:10.1103/PhysRevD.105.045018}

\bibitem{Deffayet:2023wdg}
Deffayet, C., Held, A., Mukohyama, S., Vikman, A.: Global and local stability for ghosts coupled to positive energy degrees of freedom.
JCAP \textbf{11}, 031 (2023)
\href{https://doi.org/10.1088/1475-7516/2023/11/031}{doi:10.1088/1475-7516/2023/11/031}

\bibitem{ErrastiDiez:2024hfq}
Errasti D\'\i{}ez, V., Gaset Rif\`a, J., Staudt, G.: Foundations of Ghost Stability.
Fortsch. Phys. \textbf{73} no.4, 2400268  (2025)
\href{https://doi.org/10.1002/prop.202400268}
{10.1002/prop.202400268}

\bibitem{FGS-2003}
Fiorani, E., Giachetta, G., Sadanashvily, G.: The Liouville–Arnold–Nekhoroshev theorem for non-compact invariant manifolds.
 J. Phys. A: Math. Gen. \textbf{36}, 2069 (2003)
\href{https://doi.org/10.1088/0305-4470/36/7/102}{doi:10.1088/0305-4470/36/7/102}

\bibitem{Fring:2024brg}
Fring,~A., Taira, T., Turner, B.: Higher Time-Derivative Theories from Space\textendash{}Time Interchanged Integrable Field Theories.
Universe \textbf{10}, no.5, 198 (2024)
\href{https://doi.org/10.3390/universe10050198}
{doi:10.3390/universe10050198}

\bibitem{Gyfto}
Gyftopoulos, E.P.: Lagrange Stability by Liapunov's Direct Method.
Proc. Sym. Rea. Kin. Con. TID-7662, 227 (1963) 
\href{https://elias-gyftopoulos-memorial-collection.unibs.it/otherpublications.htm}{https://elias-gyftopoulos-memorial-collection.unibs.it/otherpublications.htm}

\bibitem{Hawking:2001yt}
Hawking, S.W., Hertog, T.: Living with ghosts.
Phys. Rev. D \textbf{65}, 103515 (2002)
\href{https://doi.org/10.1103/PhysRevD.65.103515}
{doi:10.1103/PhysRevD.65.103515}

\bibitem{Kaparulin:2014vpa}
Kaparulin, D.S.,Lyakhovich, S.L., Sharapov, A.A.: Classical and quantum stability of higher-derivative dynamics.
Eur. Phys. J. C \textbf{74}, no.10, 3072 (2014)
\href{https://doi.org/10.1140/epjc/s10052-014-3072-3}
{doi:10.1140/epjc/s10052-014-3072-3}

\bibitem{Lee:topo}
Lee,~J.~M.: Introduction to Topological Manifolds. 2nd edn. Springer, New York (2011)
\href{https://doi.org/10.1007/978-1-4419-7940-7}{doi:10.1007/978-1-4419-7940-7}

\bibitem{Lee}
Lee,~J.~M.: Introduction to Smooth Manifolds. 2nd edn. Springer, New York (2012)
\href{https://doi.org/10.1007/978-1-4419-9982-5}{doi.org/10.1007/978-1-4419-9982-5}

\bibitem{PU}
Pais, A., Uhlenbeck, G.E.: On Field Theories with Non-Localized Action. Phys. Rev. \textbf{79}, 145
(1950) 
\href{https://doi.org/10.1103/PhysRev.79.145}
{doi:10.1103/PhysRev.79.145}

\bibitem{Smilga:2008pr}
Smilga, A.V.: Comments on the dynamics of the Pais-Uhlenbeck oscillator.
SIGMA \textbf{5}, 017 (2009)
\href{https://doi.org/10.3842/Sigma.2009.017}
{doi:10.3842/Sigma.2009.017}

\bibitem{Smilga:2020elp}
Smilga, A.V.: On exactly solvable ghost-ridden systems.
Phys. Lett. A \textbf{389}, 127104 (2021)
\href{https://doi.org/10.1016/j.physleta.2020.127104}
{doi:10.1016/j.physleta.2020.127104}

\end{thebibliography}
%

\end{document}